\documentclass{article}

\usepackage[english]{babel}

\usepackage[letterpaper,top=2cm,bottom=2cm,left=3cm,right=3cm,marginparwidth=1.75cm]{geometry}

\usepackage{graphicx}
\usepackage[colorlinks=true, allcolors=blue]{hyperref}
\usepackage{subcaption}
\usepackage{array}
\usepackage{multirow}
\usepackage{diagbox}
\usepackage{xcolor}
\usepackage{amsmath,amsfonts,amssymb, amsthm}
\usepackage{mathtools} 
\usepackage{parskip}
\usepackage{comment}
\usepackage{tikz}
\usetikzlibrary{positioning,fit,calc}
\usepackage{subcaption}
\usetikzlibrary{positioning,fit,backgrounds,calc,matrix}

\newenvironment{np-problem}[3]{%
\medskip
 \par\noindent {\bfseries #1}\\[0.5em]
 \begin{tabular}{@{}p{1.6cm}p{\dimexpr\linewidth-3cm-1em}@{}}
  \textit{Instance:} & #2 \\
   \textit{Question:} & #3 \\ \medskip
 \end{tabular}
 \par
}{}

\newenvironment{np-hard-problem}[3]{%
\medskip
 \par\noindent {\bfseries #1}\\[0.5em]
 \begin{tabular}{@{}p{1.6cm}p{\dimexpr\linewidth-3cm-1em}@{}}
  \textit{Instance:} & #2 \\
   \textit{Required:} & #3 \\ \medskip
 \end{tabular}
 \par
}{}

\newtheorem{lemma}{Lemma}
\newtheorem{theorem}{Theorem}
\newtheorem{claim}{Claim}
\newtheorem{definition}{Definition}

\newcommand{\dist}{\operatorname{dist}}
\newcommand{\maxdcc}[2]{$(#1,#2)$\textsc{-MAX-DCC}}
\newcommand{\pc}[2]{$(#1,#2)$\textsc{-PC}}

\title{Disjoint covering of bipartite graphs with $s$-clubs}
\author{Angelo Monti\inst{1}\orcidID{0000-0002-3309-8249} \and
Blerina Sinaimeri\inst{2}\orcidID{0000-0002-9797-7592} }

\author{A. Monti\footnote{Computer Science Department, Sapienza University of Rome, Italy,
\href{mailto:monti@di.uniroma1.it}{monti@di.uniroma1.it}}
 \and B. Sinaimeri\footnote{Luiss University, Rome, Italy. \href{mailto:bsinaimeri@luiss.it}{bsinaimeri@luiss.it}}}
 
\newcommand{\mypath}{\mathrel{\text{\textemdash\!\textemdash}}}

\begin{document}
\maketitle
\begin{abstract}
For a positive integer $s$, an $s$-club in a graph $G$ is a set of vertices inducing a subgraph with diameter at most $s$. As generalizations of cliques, $s$-clubs offer a flexible model for real-world networks. This paper addresses the problems of partitioning and disjoint covering of vertices with $s$-clubs on bipartite graphs. First we consider the \pc{k}{s} problem where ask whether the vertices of $G$ can be partitioned into at most $k$ disjoint $s$-clubs. We prove that for any fixed $k \geq 2$ and for any fixed odd $s \geq 3$ or even $s\geq 8$, the \pc{k}{s} problem  is NP-complete even for bipartite graphs. Note that our NP-completeness result is stronger than the one in Abbas and Stewart (1999), as we assume that both $s$ and $k$ are  constants and not part of the input.

Additionally, we study the Maximum Disjoint $(t,s)$-Club Covering problem (\maxdcc{t}{s}), which aims to find a collection of vertex-disjoint $(t,s)$-clubs (i.e. $s$-clubs with at least $t$ vertices) that covers the maximum number of vertices in $G$. We  prove that it is NP-hard to achieve an  approximation factor of $\frac{95}{94} $ for \maxdcc{t}{3} for any fixed $t\geq 8$
and for \maxdcc{t}{2} for any fixed $t\geq 5$ even for bipartite graphs.  Previously, results were known only for \maxdcc{3}{2}. Finally, we provide a polynomial-time algorithm for \maxdcc{2}{2} resolving an open problem from Dondi \textit{et al.} (2019).
\end{abstract}

\section{Introduction}
For a positive integer $s$, an $s$-club in a graph $G$ is a set of vertices that induces a subgraph of $G$ of diameter at most $s$.  Clubs are generalizations of cliques ($1$-clubs are exactly cliques) and offer a wider and more practical way to model real-world interactions \cite{Mokken1979,Mokken2016,Laan2016,Komusiewicz2016,Zoppis2018,ZOU2018}. Partitioning a graph into cliques is important for clustering and community detection. Consequently, there has been research into partitioning graphs into $s$-clubs as a way to extend these methods to more flexible and realistic groupings. In this paper we focus exclusively on bipartite graphs. Bipartite graphs are of particular interest due to their wide range of applications in various fields such as scheduling, matching problems, recommendation systems and network flow optimization \cite{Asratian1998}. 
They also provide a particularly interesting setting for 
$s$-club partitioning and covering. Indeed, since bipartite graphs contain no triangles, clique-based notions of cohesion are limited in this setting, whereas $s$-clubs can still capture tightly connected groups through short alternating paths. This is especially relevant in bipartite networks, where preserving the underlying two-mode structure is often important for community detection and network analysis \cite{Barber2007, Calderer2021}. Thus, $s$-clubs provide a natural framework for modeling cohesive substructures in bipartite graphs.

We examine two closely related problems involving the partitioning and covering of a graph's vertices using $s$-clubs. 
The first problem is the Minimum Partition $s$-Club problem, where the objective is to partition the vertices of a graph $G$ into the minimum number of disjoint $s$-clubs.  This problem is NP-hard for $s\geq 2$ \cite{Deogun1997}, even when restricted to bipartite or chordal graphs \cite{Abbas1999,Chang2014}.  
Clearly from these results we have that the decision version of this problem, where for a fixed $s$, given a graph $G$ and an integer $k$, we decide whether it is possible to partition the vertices of $G$ into at most $k$ disjoint $s$-clubs, is NP-complete. However, this does not address the complexity of the decision problem when $k$ is also fixed (that is both $s$ and $k$ are not part of the input). We thus, consider the problem below.

\begin{np-problem}
    {$k$-Partition $s$-Club problem (\pc{k}{s})}{A graph $G=(V,E)$.}{Is there a partition of $V$ into at most $k$ vertex disjoint $s$-clubs?}
\end{np-problem}

Previous studies have explored the complexity of this problem for some fixed values of $s$ and $k$. For $k=1$, the problem is equivalent to  determining the diameter of a graph and thus is trivially solvable in polynomial time. For $s=1$ the problem is equivalent to determine whether there exists a  partition of the vertices into $k$ cliques, that is equivalent to determine whether there exists a $k$-coloring of the complement graph. This problem is NP-complete for any $k \geq 3$ \cite{GareyJ79} and polynomial for $k=2$. More recently, \cite{BAZGAN2025} it was shown that when we restrict to split graphs the problem  \pc{2}{2} is NP-complete.  When we restrict to bipartite graphs, $s=1$ corresponds to the problem of determining whether there exists a perfect matching of size $k$ in a graph and is thus polynomial. Again for bipartite graphs and $s=2$ and $k=2$ the problem is polynomial \footnote{Notice that if we do not require the $s$-clubs to be disjoint, then the case $s=2$, $k=2$ is shown to be NP-complete for general graphs \cite{DondiL23}.}\cite{Fleischner2009}. This is a distinction with the split graphs where already \pc{2}{2} problem is NP-complete.
In this paper we show that for bipartite graphs and any fixed odd $s\geq 3$  or even  $s\geq 8$ and for any fixed $k\geq 2$  the \pc{k}{s} problem is NP-complete. Note that our NP-completeness result is stronger than the one in \cite{Abbas1999} as we assume that both $s$ and $k$ are  constants and not part of the input.

In some real-world applications, it may not be feasible to partition all vertices into $s$-clubs.  To address such cases, a variant of this problem, known as the Maximum Disjoint $(t,s)$-Club Covering problem (\maxdcc{t}{s}), was introduced by Dondi \emph{et al.} in \cite{Dondi2019}. This problem seeks to find, given a graph $G$, a collection of disjoint $(t,s)$-clubs that covers the maximum number of vertices in $G$.  A $(t,s)$-club is an $s$-club with at least $t$ vertices. The concept of $(t,s)$-clubs extends that of $s$-clubs by adding a minimum size constraint, and is motivated by applications where identifying large, well-connected subgraphs is important (see, e.g., \cite{Laan2016,Dondi2019}). Hence, the second problem we consider is the following: \\

\noindent\begin{minipage}{\textwidth}
\begin{np-hard-problem}
    {Maximum Disjoint $(t,s)$-Club Covering problem (\maxdcc{t}{s})}{A  graph $G=(V,E)$.}{A collection of vertex disjoint $(t,s)$-clubs that covers the maximum number of vertices in $V$. }
\end{np-hard-problem}
\end{minipage}

To the best of our knowledge the only cases considered in literature are the cases $s=2$ and $s=3$ \cite{Dondi2019}. In particular, for $s=2$ in \cite{Dondi2019} it is proved that \maxdcc{3}{2} is APX-hard and the case \maxdcc{2}{2} is left open. Here we show that \maxdcc{t}{2} is APX-hard for any fixed $t\geq 5$ even for bipartite graphs and \maxdcc{2}{2} can be solved in polynomial time for general graphs. For $s=3$ in \cite{Dondi2019} it is shown that \maxdcc{2}{3} is polynomial and the case \maxdcc{3}{3} is left as an open problem. We  show that \maxdcc{t}{3} is APX-hard for any fixed $t \geq 8$, even for bipartite graphs.

The paper is organized as follows: In Section~\ref{sec:preliminaries} we introduce the definitions we need for the paper. In Section~\ref{sec:partition-npc} we show that for any fixed $k\geq 2$ and for any fixed $s \geq 3 $ with $s\neq 4, s\neq 6$  the \pc{k}{s} problem is NP-complete.  In Section~\ref{sec:hardness-max-dcc} we prove that it is NP-hard to achieve an  approximation factor of $\frac{95}{94} $ for\maxdcc{t}{3} for any fixed $t\geq 8$ and for \maxdcc{t}{2} for any fixed $t\geq 5$ even for bipartite graphs. On the positive side we provide a polynomial-time algorithm for \maxdcc{2}{2}. Finally, in Section~\ref{sec:conslusions} we conclude with some open problems.

\section{Preliminaries}\label{sec:preliminaries}
All the graphs we consider here are undirected and simple (with no loops or multiple edges). For a graph $G=(V,E)$ and a subset $V' \subseteq V$ we denote by $G[V']$ the subgraph induced by the vertices in $V'$.  For  any two vertices $u, v \in V$ we denote by $u\mypath v$ a path connecting $u$ and $v$ in $G$. To clarify notation and avoid confusion, we will slightly abuse notation by writing $(u,v)$ for the edge between $u$ and $v$, instead of $\{u,v\}$. A subset $V'$ of vertices is called an \textit{$s$-club} if the diameter of $G[V']$ is at most $s$. In other words, every pair of vertices in the $s$-club can be connected by a path of length at most $s$ within the subgraph. An \textit{$(t,s)$-club} is an $s$-club of at least $t$ vertices.

The \textit{closed neighborhood} of a vertex $v$ in a graph $G=(V,E)$, denoted by $N[v]$, is the set consisting of the vertex $v$ itself and all vertices adjacent to $v$. Formally,  $ N[v] = \{v\} \cup \{u \in V \mid (u, v) \in E\}$.

A graph $G = (V, E)$ is \textit{bipartite} if its vertices can be partitioned into two independent sets. We denote it as $G = (V_1, V_2, E)$, where $V_1$ and $V_2$ are the independent sets.
A tree $T_G$ is said to be a \textit{spanning tree} of a connected graph $G$ if $T_G$ is a subgraph of $G$ (not necessarily induced) and $T_G$ contains all vertices of $G$. A \textit{rooted tree} is a tree with a special vertex labelled as the \textit{root}  of the  tree. In a tree, a vertex $v$ is said to be at \textit{level} $l$ if $v$ is at a distance $l$ from the root. The \textit{height} of a tree is the maximum level which occurs in the tree. The \textit{parent} of a vertex $v$ is the vertex connected to $v$ on the path to the root.

For any integer $k$ we denote by $[k]$ the set $\{1,2, \ldots, k\}$. We denote by $2^{[k]}$ the family of all possible subsets of $[k]$.

\section{NP-completeness of \pc{k}{s} problem}\label{sec:partition-npc}
It is claimed in \cite{Fleischner2009} that \pc{k}{2} is NP-complete for every fixed $k \geq 3$, via a reduction from the $k$-List Coloring problem. However, the correctness proof of the reduction contains a flaw. While the construction shows that a valid $k$-list coloring of the original graph yields a partition into at most $k$ subgraphs of diameter at most $2$, the converse implication is not fully justified. In particular, the argument that the existence of a $k$-partition into $2$-clubs in the constructed graph implies the existence of a corresponding $k$-list coloring of the original graph does not hold in general. In this section we prove the NP-completeness of \pc{k}{s} for every $k\geq 2$ and odd $s\geq 3$ and even $s\geq 8$. We start by proving that \pc{2}{3} is NP-complete even for bipartite graphs. We will reduce from the Monotone $3$-SAT problem  \cite{GareyJ79}. The reduction is inspired by the  proof in \cite{BAZGAN2025} where it is shown  that the \pc{2}{2} problem is NP-complete for split graphs.

\begin{np-problem}
    {Monotone $3$-SAT}{A set $X$ of variables, a collection $C$ of clauses over $X$, which contains either only negated variables or only positive variables and such that for each clause $c \in C$, $|c|=3$. }{Is there a satisfying truth assignment for $C$?}
\end{np-problem}

\begin{lemma}\label{lem:PC(2,3)}
The \pc{2}{3} problem is NP-complete even when restricted to bipartite graphs.
\end{lemma}

\begin{proof}
Membership in NP is immediate: given a partition $(X,Y)$ of $V(G)$, one can check in polynomial time whether both $G[X]$ and $G[Y]$ have diameter at most~$3$.

For NP-hardness, we reduce from \textsc{Monotone 3-SAT}. Let
$\phi$ be an instance of \textsc{Monotone 3-SAT} and let 
\[
X(\phi)=\{x_1,\dots,x_n\}
\qquad\text{and}\qquad
C(\phi)=\{c_1,\dots,c_m\}
\]
be the sets of variables and clauses, respectively. We write
\[
C^+ := \{c_i\in C(\phi)\mid c_i \text{ is positive}\},
\qquad
C^- := \{c_i\in C(\phi)\mid c_i \text{ is negative}\}.
\]

Notice that we may assume that $|C^+|\ge 2$ and $|C^-|\ge 2$.
Indeed, if all clauses have the same sign, then $\phi$ is trivially satisfiable:
if all clauses are positive we set every variable to \textsc{True}, and if all
clauses are negative we set every variable to \textsc{False}. Otherwise, if
$\phi$ contains exactly one positive clause (respectively exactly one negative
clause), we simply duplicate that clause. This does not affect satisfiability,
and after this modification we have $|C^+|\ge 2$ and $|C^-|\ge 2$.
From $\phi$ we construct a graph $G=G(\phi)$ as follows.

\medskip
\noindent\textbf{Vertices.}
\begin{itemize}
\item For each variable $x_i$, create a vertex $v_i$.
Let
\[
V_{\mathrm{var}}:=\{v_1,\dots,v_n\}.
\]

\item For each positive clause $c_i\in C^+$, create two vertices
$a_i^+$ and $b_i^+$; for each negative clause $c_i\in C^-$, create two vertices
$a_i^-$ and $b_i^-$. Let
\[
Q_i^+ := \{a_i^+,b_i^+\}, \qquad Q_i^- := \{a_i^-,b_i^-\},
\]
\[
V_{\mathrm{cl}}^+ := \bigcup_{c_i\in C^+} Q_i^+,
\qquad
V_{\mathrm{cl}}^- := \bigcup_{c_i\in C^-} Q_i^-,
\qquad
V_{\mathrm{cl}} := V_{\mathrm{cl}}^+ \cup V_{\mathrm{cl}}^-.
\]

\item For every pair of distinct positive clauses $c_i,c_j\in C^+$ with $i<j$,
create four vertices
\[
P_{i,j}^+ := \{p_{a_i a_j}^+,\, p_{a_i b_j}^+,\, p_{b_i a_j}^+,\, p_{b_i b_j}^+\}.
\]
Similarly, for every pair of distinct negative clauses $c_i,c_j\in C^-$ with $i<j$,
create four vertices
\[
P_{i,j}^- := \{p_{a_i a_j}^-,\, p_{a_i b_j}^-,\, p_{b_i a_j}^-,\, p_{b_i b_j}^-\}.
\]
Let
\[
V_{\mathrm{pair}}^+ := \bigcup_{\substack{c_i,c_j\in C^+\\ i<j}} P_{i,j}^+,
\qquad
V_{\mathrm{pair}}^- := \bigcup_{\substack{c_i,c_j\in C^-\\ i<j}} P_{i,j}^-,
\qquad
V_{\mathrm{pair}} := V_{\mathrm{pair}}^+ \cup V_{\mathrm{pair}}^-.
\]

\item Add four special vertices $u^+,u^-,v^+,v^-$, and let
\[
V_{\mathrm{sp}}:=\{u^+,u^-,v^+,v^-\}.
\]

\end{itemize}

\medskip
\noindent\textbf{Edges.}
\begin{itemize}
\item Suppose a clause $c_i$ contains the variables $x_j,x_\ell,x_r$.
If $c_i\in C^+$, then connect each of $a_i^+,b_i^+$ to each of
$v_j,v_\ell,v_r$.
If $c_i\in C^-$, then connect each of $a_i^-,b_i^-$ to each of
$v_j,v_\ell,v_r$.
Let $E_{\mathrm{cv}}$ denote the set of all these clause-variable edges.

\item For every pair $c_i,c_j\in C^+$ with $i<j$, connect each pair-vertex to
the two clause vertices indicated by its subscript:
\[
\begin{aligned}
p_{a_i a_j}^+a_i^+,\quad p_{a_i a_j}^+a_j^+,\qquad
p_{a_i b_j}^+a_i^+,\quad p_{a_i b_j}^+b_j^+,\\
p_{b_i a_j}^+b_i^+,\quad p_{b_i a_j}^+a_j^+,\qquad
p_{b_i b_j}^+b_i^+,\quad p_{b_i b_j}^+b_j^+ .
\end{aligned}
\]
Define the edges for $P_{i,j}^-$ analogously. Let $E_{\mathrm{pc}}$ be the union
of all these pair-clause edges.

\item Connect $v^+$ to every vertex of $V_{\mathrm{pair}}^+$, and connect $v^-$
to every vertex of $V_{\mathrm{pair}}^-$.

\item Connect $u^+$ to every vertex of $V_{\mathrm{pair}}^+\cup V_{\mathrm{var}}$,
and connect $u^-$ to every vertex of $V_{\mathrm{pair}}^-\cup V_{\mathrm{var}}$.
\end{itemize}

Finally, set
\[
V(G)=V_{\mathrm{var}}\cup V_{\mathrm{cl}}\cup V_{\mathrm{pair}}\cup V_{\mathrm{sp}},
\]
and let $E(G)$ be the union of all edges defined above.

Figure~\ref{fig:gadget-pc2-3} illustrates the construction
for a small instance. 

\begin{figure}[t]
\centering
\begin{tikzpicture}[
  node distance=6mm and 18mm,
  every node/.style={font=\small},
  var/.style={circle, draw, fill=orange!18, minimum size=16pt, inner sep=1pt},
  clp/.style={circle, draw, fill=blue!10, minimum size=16pt, inner sep=1pt},
  cln/.style={circle, draw, fill=blue!10, minimum size=16pt, inner sep=1pt},
  pp/.style={circle, draw, fill=orange!18, minimum size=18pt, inner sep=1pt},
  pn/.style={circle, draw, fill=orange!18, minimum size=18pt, inner sep=1pt},
  sp/.style={rectangle, draw, rounded corners=2pt, fill=blue!10, inner sep=2pt},
  e1/.style={draw, line width=0.5pt, opacity=0.9},
  e2/.style={draw, line width=0.4pt, opacity=0.55},
  e3/.style={draw, dashed, line width=1pt, opacity=0.20}
]

\matrix (Clauses) [
  matrix of nodes,
  nodes={anchor=center},
  row sep=5mm
] {
  \node[clp] (a1p) {$a_1^{+}$};\\
  \node[clp] (b1p) {$b_1^{+}$};\\
  \node[clp] (a2p) {$a_2^{+}$};\\
  \node[clp] (b2p) {$b_2^{+}$};\\[2mm]
  \node[cln] (a3n) {$a_3^{-}$};\\
  \node[cln] (b3n) {$b_3^{-}$};\\
  \node[cln] (a4n) {$a_4^{-}$};\\
  \node[cln] (b4n) {$b_4^{-}$};\\
};

\matrix (Pairs) [
  right=24mm of Clauses,
  matrix of nodes,
  nodes={anchor=center},
  row sep=5mm
] {
  \node[pp] (pAAp) {$p^{+}_{a_1a_2}$};\\
  \node[pp] (pBAp) {$p^{+}_{b_1a_2}$};\\
  \node[pp] (pABp) {$p^{+}_{a_1b_2}$};\\
  \node[pp] (pBBp) {$p^{+}_{b_1b_2}$};\\[2mm]
  \node[pn] (pAAn) {$p^{-}_{a_3a_4}$};\\
  \node[pn] (pBAn) {$p^{-}_{b_3a_4}$};\\
  \node[pn] (pABn) {$p^{-}_{a_3b_4}$};\\
  \node[pn] (pBBn) {$p^{-}_{b_3b_4}$};\\
};

\matrix (Specials) [
  right=20mm of Pairs,
  matrix of nodes,
  nodes={anchor=center},
  row sep=10mm
] {
  \node[sp] (vp) {$v^{+}$};\\
  \node[sp] (up) {$u^{+}$};\\[4mm]
  \node[sp] (vm) {$v^{-}$};\\
  \node[sp] (um) {$u^{-}$};\\
};
\matrix (Vars) [
  right=22mm of Specials,
  matrix of nodes,
  nodes={anchor=center},
  row sep=5mm
] {
  \node[var] (v1) {$v_1$};\\
  \node[var] (v2) {$v_2$};\\
  \node[var] (v3) {$v_3$};\\
  \node[var] (v4) {$v_4$};\\
  \node[var] (v5) {$v_5$};\\
  \node[var] (v6) {$v_6$};\\
};

\foreach \vv in {v1,v2,v3}{
  \draw[e3] (a1p) -- (\vv);
  \draw[e3] (b1p) -- (\vv);
}
\foreach \vv in {v2,v4,v5}{
  \draw[e3] (a2p) -- (\vv);
  \draw[e3] (b2p) -- (\vv);
}
\foreach \vv in {v1,v4,v6}{
  \draw[e3] (a3n) -- (\vv);
  \draw[e3] (b3n) -- (\vv);
}
\foreach \vv in {v3,v5,v6}{
  \draw[e3] (a4n) -- (\vv);
  \draw[e3] (b4n) -- (\vv);
}

\draw[e1] (pAAp) -- (a1p); \draw[e1] (pAAp) -- (a2p);
\draw[e1] (pABp) -- (a1p); \draw[e1] (pABp) -- (b2p);
\draw[e1] (pBAp) -- (b1p); \draw[e1] (pBAp) -- (a2p);
\draw[e1] (pBBp) -- (b1p); \draw[e1] (pBBp) -- (b2p);

\draw[e1] (pAAn) -- (a3n); \draw[e1] (pAAn) -- (a4n);
\draw[e1] (pABn) -- (a3n); \draw[e1] (pABn) -- (b4n);
\draw[e1] (pBAn) -- (b3n); \draw[e1] (pBAn) -- (a4n);
\draw[e1] (pBBn) -- (b3n); \draw[e1] (pBBn) -- (b4n);

\foreach \ppv in {pAAp,pABp,pBAp,pBBp}{
  \draw[e2] (vp) -- (\ppv);
  \draw[e2] (up) -- (\ppv);
}
\foreach \pnv in {pAAn,pABn,pBAn,pBBn}{
  \draw[e2] (vm) -- (\pnv);
  \draw[e2] (um) -- (\pnv);
}

\draw[e2] (up) to[out=10,in=180]   (v1);
\draw[e2] (up) to[out=6,in=180]    (v2);
\draw[e2] (up) to[out=2,in=180]    (v3);
\draw[e2] (up) to[out=-2,in=180]   (v4);
\draw[e2] (up) to[out=-6,in=180]   (v5);
\draw[e2] (up) to[out=-10,in=180]  (v6);

\draw[e2] (um) to[out=12,in=180]   (v1);
\draw[e2] (um) to[out=8,in=180]    (v2);
\draw[e2] (um) to[out=4,in=180]    (v3);
\draw[e2] (um) to[out=0,in=180]    (v4);
\draw[e2] (um) to[out=-4,in=180]   (v5);
\draw[e2] (um) to[out=-8,in=180]   (v6);

\begin{scope}[on background layer]
  \node[
    draw, rounded corners=3pt, inner sep=6pt, fit=(Clauses),
    label={[font=\small]above:{$V_{\mathrm{cl}}$}}
  ] {};
  \node[
    draw, rounded corners=3pt, inner sep=6pt, fit=(Pairs),
    label={[font=\small]above:{$V_{\mathrm{pair}}$}}
  ] {};
  \node[
    draw, rounded corners=3pt, inner sep=6pt, fit=(Specials),
    label={[font=\small]above:{$V_{\mathrm{sp}}$}}
  ] {};
  \node[
    draw, rounded corners=3pt, inner sep=6pt, fit=(Vars),
    label={[font=\small]above:{$V_{\mathrm{var}}$}}
  ] {};
\end{scope}

\end{tikzpicture}
\caption{An example of the graph $G(\phi)$ for
$\phi=(x_1 \vee x_2 \vee x_3)\wedge(x_2 \vee x_4 \vee x_5)\wedge
(\overline{x_1}\vee\overline{x_4}\vee\overline{x_6})\wedge
(\overline{x_3}\vee\overline{x_5}\vee\overline{x_6})$.
For readability, the clause-variable edges are drawn as dashed lines. The two colors of the vertices indicate the bipartition of $G(\phi)$.}
\label{fig:gadget-pc2-3}
\end{figure}

By construction, $G$ is bipartite with bipartition
\[
L:=V_{\mathrm{var}}\cup V_{\mathrm{pair}}
\qquad\text{and}\qquad
R:=V_{\mathrm{cl}}\cup V_{\mathrm{sp}}.
\]

We now prove that $\phi$ is satisfiable if and only if $V(G)$ can be partitioned
into two sets whose induced subgraphs both have diameter at most~$3$.

\medskip
\noindent $(\Rightarrow)$
Assume that $\phi$ is satisfiable, and let $\tau$ be a satisfying truth assignment.

We partition the variable vertices according to $\tau$:
\[
V_{\mathrm{var}}^+ := \{v_i\mid \tau(x_i)=\textsc{True}\},
\qquad
V_{\mathrm{var}}^- := \{v_i\mid \tau(x_i)=\textsc{False}\}.
\]
We then define
\[
X := V_{\mathrm{var}}^+ \cup V_{\mathrm{cl}}^+ \cup V_{\mathrm{pair}}^+ \cup \{u^+,v^+\},
\]
\[
Y := V_{\mathrm{var}}^- \cup V_{\mathrm{cl}}^- \cup V_{\mathrm{pair}}^- \cup \{u^-,v^-\}.
\]
Clearly, $(X,Y)$ is a partition of $V(G)$.

We show that $G[X]$ has diameter at most~$3$; the argument for $G[Y]$ is symmetric.

Observe first that in $G[X]$:
\begin{itemize}
\item $u^+$ is adjacent to every vertex of $V_{\mathrm{var}}^+\cup V_{\mathrm{pair}}^+$;
\item $v^+$ is adjacent to every vertex of $V_{\mathrm{pair}}^+$.
\end{itemize}
Hence:
\begin{itemize}
\item any two vertices of $V_{\mathrm{var}}^+$ are at distance at most $2$ via $u^+$;
\item any vertex of $V_{\mathrm{var}}^+$ and any vertex of $V_{\mathrm{pair}}^+$ are at distance at most $2$ via $u^+$;
\item any two vertices of $V_{\mathrm{pair}}^+$ are at distance at most $2$ via either $u^+$ or $v^+$;
\item $u^+$ and $v^+$ are at distance $2$ through any vertex of $V_{\mathrm{pair}}^+$.
\end{itemize}

It remains to consider the vertices in $V_{\mathrm{cl}}^+$.
Let $c_i$ be a positive clause and let $z \in \{a_i^+, b_i^+\}$ be
one of its clause vertices. Since the assignment $\tau$ satisfies $\phi$,
the clause $c_i$ contains some variable $x_j$ with
$\tau(x_j)=\textsc{True}$. Hence $v_j \in V_{\mathrm{var}}^+$ and,
by construction of $G$, the vertex $z$ is adjacent to $v_j$.
Therefore
\[
u^+ - v_j - z
\]
is a path of length $2$ in $G[X]$. Thus every vertex of
$V_{\mathrm{cl}}^+$ is at distance at most $2$ from $u^+$.

Moreover, recall that we assumed $|C^+|\ge 2$. Hence, besides $c_i$,
there exists another positive clause. By construction, this implies that
there is a vertex $p \in V_{\mathrm{pair}}^+$ adjacent to $z$.
Since $v^+$ is adjacent to every vertex of $V_{\mathrm{pair}}^+$, we obtain
a path
\[
z - p - v^+
\]
of length $2$ in $G[X]$. Therefore every vertex of $V_{\mathrm{cl}}^+$
is also at distance at most $2$ from $v^+$. It follows that every vertex of $V_{\mathrm{cl}}^+$ is at distance at most $3$
from every vertex of $V_{\mathrm{var}}^+ \cup V_{\mathrm{pair}}^+ \cup \{u^+,v^+\}$.

It remains to consider distances between two vertices of $V_{\mathrm{cl}}^+$.
Let $z,z' \in V_{\mathrm{cl}}^+$.
If $z \in Q_i^+$ and $z' \in Q_j^+$ with $i \neq j$, then by construction
there is a pair vertex in $P_{i,j}^+$ adjacent to both $z$ and $z'$.
Hence
\[
\operatorname{dist}_{G[X]}(z,z') \le 2.
\]

If instead $z,z' \in Q_i^+$, then $z$ and $z'$ are the two clause vertices
of the same positive clause $c_i$. Since $\tau$ satisfies $c_i$, one of the
variables of $c_i$ is set to \textsc{True}; let $x_j$ be such a variable.
Then $v_j \in V_{\mathrm{var}}^+$, and both $z$ and $z'$ are adjacent to $v_j$.
Therefore
\[
z - v_j - z'
\]
is a path of length $2$ in $G[X]$, and thus
\[
\operatorname{dist}_{G[X]}(z,z') \le 2.
\]

Consequently, every two vertices of $G[X]$ are at distance at most $3$, and hence
\[
\operatorname{diam}(G[X]) \le 3.
\]

By symmetry, $\operatorname{diam}(G[Y])\le 3$.

\medskip
\noindent $(\Leftarrow)$
Suppose that $V(G)$ can be partitioned into two sets $(X,Y)$ such that both
$G[X]$ and $G[Y]$ have diameter at most~$3$.
Without loss of generality, assume that $v^+\in X$.

\begin{claim}
$v^-$ belongs to $Y$.
\end{claim}

\begin{proof}
The neighbors of $v^+$ are exactly the vertices of $V_{\mathrm{pair}}^+$, and the
neighbors of $v^-$ are exactly the vertices of $V_{\mathrm{pair}}^-$. Since there are no
edges between $V_{\mathrm{pair}}^+$ and $V_{\mathrm{pair}}^-$, every path from $v^+$ to $v^-$ must pass through at least another vertex.
Hence every such path has length strictly greater than  $3$.   Therefore
$\operatorname{dist}_G(v^+,v^-)\ge 4$, so $v^+$ and $v^-$ cannot lie in the same part of a
$3$-club partition. Since $v^+\in X$, we obtain $v^-\in Y$.
\end{proof}

\begin{claim}
All vertices of $V_{\mathrm{cl}}^+$ belong to $X$.
\end{claim}

\begin{proof}
Suppose that some vertex $z\in V_{\mathrm{cl}}^+$ belongs to $Y$. Since $v^-\in Y$
and $G[Y]$ has diameter at most $3$, we must have $\dist_{G[Y]}(z,v^-)\le 3$.
But this is impossible: from a vertex of $V_{\mathrm{cl}}^+$ one can only move to
a variable vertex or to a vertex of $V_{\mathrm{pair}}^+$, while the only neighbors
of $v^-$ are the vertices of $V_{\mathrm{pair}}^-$ and moreover there exists no edge connecting a vertex in $V_{\mathrm{pair}}^+ \cup V_{\mathrm{cl}}^+$ and one in $V_{\mathrm{pair}}^-$. Hence every path from $z$ to $v^-$
has length strictly greater than $3$. This contradiction shows that $z\notin Y$, and therefore $z\in X$.
\end{proof}

\begin{claim}
Let $c_i\in C^+$ be a positive clause, and let $x,y,z$ be the variable vertices
corresponding to the three variables of $c_i$. Then at least one of $x,y,z$
belongs to $X$.
\end{claim}

\begin{proof}
Let $a_i^+,b_i^+$ be the two clause vertices representing $c_i$. By the previous
claim, both $a_i^+$ and $b_i^+$ belong to $X$. Now the only common neighbors of
$a_i^+$ and $b_i^+$ are precisely the three variable vertices $x,y,z$. If all
three belonged to $Y$, then in $G[X]$ the vertices $a_i^+$ and $b_i^+$ would have
no path of length at most $2$, and in particular they could not be within distance
at most $3$, contradicting that $G[X]$ has diameter at most $3$. Hence at least
one of $x,y,z$ lies in $X$.
\end{proof}

Now define the truth assignment $\tau$ by setting
\[
\tau(x_i)=
\begin{cases}
\textsc{True} & \text{if } v_i\in X,\\
\textsc{False} & \text{if } v_i\in Y.
\end{cases}
\]
By Claim~3, every positive clause contains a variable whose vertex lies in $X$,
and hence a variable set to \textsc{True} by $\tau$. By the symmetric argument,
every negative clause contains a variable whose vertex lies in $Y$, and hence a
variable set to \textsc{False} by $\tau$. Therefore $\tau$ satisfies $\phi$.

\end{proof}

We next extend this hardness result in two steps. First, we show how to pass from partitioning in $3$-clubs to any fixed diameter $s\ge 3$ and $s\neq 4$. Then we show how to lift the number of partitions $k$ from $2$ to any fixed $k\ge 2$.

\begin{lemma}\label{lem:morediameter}
For every odd $s \ge 3$ and every even $s \ge 8$,
\pc{2}{s} is NP-complete on bipartite graphs.
\end{lemma}

\begin{proof}
Membership in NP is immediate. We know by Lemma~\ref{lem:PC(2,3)} that \pc{2}{3} is NP-complete on bipartite graphs so we will reduce from that. Let $G$ be an instance of
\pc{2}{3} on bipartite graphs and fix an integer $s\geq 3$, $s\neq 4, s\neq 6$. We construct in polynomial time a
bipartite graph $G'$ such that $G$ is a yes-instance of \pc{2}{3}
if and only if $G'$ is a yes-instance of $(2,s)$\textsc{-PC}$.$ We distinguish two cases. \\

\noindent\textbf{Case 1: $s$ is odd.}
Then $s=2q+3$ for some integer $q\ge 0$.
Construct $G'$ from $G$ by attaching to every vertex $v\in V(G)$ a pendant path of length $q$ and we denote by $R(v)$ the nodes in this path. Since adding pendant paths preserves bipartiteness, $G'$ is bipartite.

For every vertex $u\in V(G')$, define $\rho(u)\in V(G)$  as the unique  vertex $v\in V(G)$ such that either $u=v$ or $u$ lies on the path $R(v)$ attached to $v$. 

Assume first that $G$ is a yes-instance of \pc{2}{3}, and let
$(X,Y)$ be a partition of $V(G)$ into two $3$-clubs. Set
\[
X' := X\cup \bigcup_{v\in X} R(v),
\qquad
Y' := Y\cup \bigcup_{v\in Y} R(v).
\]
Clearly $(X',Y')$ is a partition of $V(G')$.

We claim that $G'[X']$ is an $s$-club. Let $u,v\in X'$. In $G'[X']$ we can go
from $u$ to $\rho(u)$ along the attached path, then from $\rho(u)$ to
$\rho(v)$ inside $G[X]$, and finally from $\rho(v)$ to $v$ along the
attached path. Hence
\[
\operatorname{dist}_{G'[X']}(u,v)
= \operatorname{dist}_{G'[X']}(u,\rho(u))+\operatorname{dist}_{G[X]}(\rho(u),\rho(v))+ \operatorname{dist}_{G'[X']}(\rho(v),v)
\le 2q+3=s.
\]
Therefore $G'[X']$ has diameter at most $s$, so it is an $s$-club.
By the same argument, $G'[Y']$ is also an $s$-club. Hence $G'$ is a
yes-instance of \pc{2}{s}.

Conversely, assume that $G'$ is a yes-instance of \pc{2}{s}, and let
$(X',Y')$ be a partition of $V(G')$ into two $s$-clubs.

For every $v\in V(G)$, let $v^\star$ denote the leaf of the path attached to $v$,
that is, the vertex of $R(v)$ farthest from $v$. Define
\[
X:=\{v\in V(G): v^\star\in X'\},
\qquad
Y:=\{v\in V(G): v^\star\in Y'\}.
\]
Since each attached leaf belongs to exactly one between $X',Y'$, then  $(X,Y)$ is a partition of $V(G)$.

We show that $G[X]$ is a $3$-club. Let $u,v\in X$, $u\neq v$. Since by definition
$u^\star,v^\star\in X'$ and $G'[X']$ is connected, every $u^\star \mypath v^\star$ path
in $G'[X']$ must pass through $u$ and $v$ (the attached paths are pendant).
Thus
\[
\operatorname{dist}_{G'[X']}(u^\star,v^\star)
=
\operatorname{dist}_{G[X]}(u,v)+2q.
\]
Because $G'[X']$ is an $s$-club and $s=2q+3$, we obtain
\[
\operatorname{dist}_{G[X]}(u,v)
=
\operatorname{dist}_{G'[X']}(u^\star,v^\star)-2q
\le s-2q = 3.
\]
Hence $G[X]$ is a $3$-club. Similarly, $G[Y]$ is a $3$-club.
Therefore $G$ is a yes-instance of \pc{2}{3}.\\

\noindent\textbf{Case 2: $s$ is even.}
Since $s\ge 8$, we have that $s=2q+6$ for some integer
$q\ge 1$. Construct $G'$ by first subdividing every edge $(u,v)$ of $G$ once, adding a node $x_{u,v}$, we denote by  $S$  the set of all such  vertices  and then, similarly as in the previous case,  attaching to every  vertex $v\in V(G)$ a proper path of length $q$. We denote by $R(v)$ the nodes in this path.   Since, both operations preserve bipartiteness $G'$ is bipartite.

Assume first that $G$ is a yes-instance of \pc{2}{3}, and let
$(X,Y)$ be a partition of $V(G)$ into two $3$-clubs.  Define
\[
S_Y:=\{x_{uv}\in S:u,v\in Y\},
\]
 Then set
\[
X' := X\cup \bigcup_{v\in X}R(v)\cup (S\setminus S_Y),
\qquad
Y' := Y\cup \bigcup_{v\in Y}R(v)\cup S_Y.
\]
Clearly $(X',Y')$ is a partition of $V(G')$. For every vertex $u\in X'$, let $\rho(u)$ be a vertex of $X$ at minimum distance from $u$ in $G'$. If there is more than one such vertex, choose $\rho(u)$ arbitrarily among them. 
Notice that 

\[
\operatorname{dist}_{G'[X']}(u,\rho(u))=
\begin{cases}
0, & \text{if } u\in X,\\[1mm]
1, & \text{if } u\in (S \setminus S_Y),\\[1mm]
\operatorname{dist}_{G'[X']}(u,v)\le q, & \text{if } u\in R(v)\text{ for some }v\in X.
\end{cases}
\]

Hence,  $\operatorname{dist}_{G'[X']}(u,\rho(u))\leq q$ as $q\geq 1$.  Moreover, for any two vertices $u, v\in X'$, we have that $\operatorname{dist}_{G'[X']}(\rho(u),\rho(v))= 2 \operatorname{dist}_{G[X]}(\rho(u),\rho(v))\leq 6$ as by construction both $\rho(u),\rho(v) \in X$ and $G[X]$ is a $3$-club.

Let $u, v\in X'$.  We have
\[
\operatorname{dist}_{G'[X']}(u,v)
\leq \operatorname{dist}_{G'[X']}(u,\rho(u))+\operatorname{dist}_{G'[X']}(\rho(u),\rho(v))+ \operatorname{dist}_{G'[X']}(\rho(v),v)
\le 2q+6=s.
\]
Therefore $G'[X']$ has diameter at most $s$, so it is an $s$-club.
By the same argument, $G'[Y']$ is also an $s$-club. Hence $G'$ is a
yes-instance of \pc{2}{s}.

\bigskip
Conversely, assume that $G'$ is a yes-instance of \pc{2}{s}, and let
$(X',Y')$ be a partition of $V(G')$ into two $s$-clubs. 

For every $v\in V(G)$, let $v^\star$ denote the leaf of the path attached to $v$,
that is, the vertex of $R(v)$ farthest from $v$. Define
\[
X:=\{v\in V(G): v^\star\in X'\},
\qquad
Y:=\{v\in V(G): v^\star\in Y'\}.
\]
Since each attached leaf belongs to exactly one between $X',Y'$, then  $(X,Y)$ is a partition of $V(G)$.

We show that $G[X]$ is a $3$-club. Let $u,v\in X$, $u\neq v$. Since by definition
$u^\star,v^\star\in X'$ and $G'[X']$ is connected, every $u^\star \mypath v^\star$ path
in $G'[X']$ must pass through $u$ and $v$ (the attached paths are pendant).
Thus
\[
\operatorname{dist}_{G'[X']}(u^\star,v^\star)
=
\operatorname{dist}_{G'[X']}(u,v)+2q= 2\operatorname{dist}_{G[X]}(u,v) +2q 
\]

\noindent
where the last equality follows by the construction of $G'$ from $G$. Because $G'[X']$ is an $s$-club and $s=2q+6$, we obtain $\operatorname{dist}_{G[X]}(x,y) \leq 3$. Hence $G[X]$ is a $3$-club. Similarly, $G[Y]$ is a $3$-club.
Therefore $G$ is a yes-instance of \pc{2}{3}.
\end{proof}

\begin{theorem}\label{theo:PC(k,s)}
For every fixed $k\ge 2$ and every fixed odd $s\ge 3$ or even $s\geq 8$,
the \pc{k}{s} problem is NP-complete.
\end{theorem}

\begin{proof}
Membership in NP is immediate.   We know by Lemma~\ref{lem:morediameter} that \pc{2}{s} is NP-complete on bipartite graphs so we will reduce from that. If disconnected input graphs are allowed, then the result follows immediately by adding $k-2$ isolated vertices, each of which must form a singleton part in any feasible partition. If one wants to keep the graph connected, we use a slightly different gadget.  

Let $G=(V,E)$ be a bipartite graph. Construct a graph $G'=(V',E')$ as follows:

\begin{itemize}
    \item Choose an arbitrary vertex $v \in V$.
    \item Attach a path $P$ of $(s+1) \cdot (k-2)$ vertices to  $v$ and denote 
    \[
    P = w_1, w_2, \dots, w_{(s+1) \cdot (k-2)}.
    \]
\end{itemize}
The resulting graph $G'$ remains bipartite, and the construction is polynomial.

\noindent
$(\Rightarrow)$ Suppose $G$ admits a $(2,s)$-partition $X,Y$ such that $G[X], G[Y]$ are $s$-clubs. In $G'$ we keep $X$ and $Y$ for the vertices of $V(G)$  and assign consecutive segments of length $s$ (i.e. of $s+1$ vertices) of the path $P$ to $k-2$ partitions.  Each segment has diameter $s$ so all induced subgraphs have diameter at most $s$.   Thus $G'$ admits a $(k,s)$-partition.

\noindent
$(\Leftarrow)$ Conversely, suppose $G'$ admits a $(k,s)$-partition $\mathcal{P}=\{V_1,\dots,V_k\}$.  Consider the path $P$ and its vertices $w_i$ with position $(s+1)i+1$, with $0\leq i \leq k-3$. Any two of these vertices are at a distance greater than $s$ and thus these $k-2$ vertices must belong to different $s$-clubs. Moreover, no vertex of $G$ can be in the same partition with any of these vertices. Indeed, the smallest distance from a vertex of $G$ to these $w_i$ is between $v$ and $w_{(s+1)(k-3)+1}$ which is $s+1$. Thus  $k-2$ sets of the partition are used for $P$, leaving  $2$ sets for the vertices of $G$.  Restricting the partition $\mathcal{P}$ to $V$ gives a $(2,s)$-partition of $G$.

Hence $G$ is a yes-instance of \pc{2}{s} if and only if $G'$ is a yes-instance of \pc{k}{s}.  Therefore  \pc{k}{s} is NP-complete on bipartite graphs
\end{proof}

\section{Hardness of \maxdcc{t}{s} problem}\label{sec:hardness-max-dcc}

\subsection{The \maxdcc{t}{3} problem}
In \cite{Dondi2019} it was proven that \maxdcc{2}{3} can be solved in polynomial time and the complexity of \maxdcc{3}{3} was posed as an open problem. In this section we prove that \maxdcc{t}{3} is APX-hard for any fixed $t \geq 8$, even in bipartite graphs. We will use the following problem, which was proven to be NP-hard to approximate in \cite{chlebik2003}.
\begin{np-hard-problem}
{Maximum $2$ Bounded $3$-Dimensional Matching problem (Max-2B3DM)}{A set $M \subseteq X \times Y \times Z$ of ordered triples where  $X$, $Y$ and $Z$ are disjoint sets and the number of
occurrences in $M$ of an element in $X$, $Y$ or $Z$ is bounded by  constant $2$.}{The largest matching $M' \subseteq M$, that is, the largest  subset such that no two elements of $M$ agree in any coordinate.}  
\end{np-hard-problem}
To prove that \maxdcc{t}{3} for \(t \geq 8\) is NP-hard, and to derive its
inapproximability bound, we give an \(L\)-reduction from Max-2B3DM.
To prove that \maxdcc{t}{3} for \(t \geq 8\) is NP-hard, and to derive its
inapproximability bound, we give an \(L\)-reduction from Max-2B3DM.

To make the paper self-contained, we recall the definition of an \(L\)-reduction
and the approximation-preserving property that we will use. For more details on
\(L\)-reductions, see \cite{williamson2011design}.

Given an instance \(I\) of an optimization problem \(A\), let \(opt(I)\) denote the
value of an optimal solution for \(I\).

\begin{definition}
Let \(A\) and \(B\) be two optimization problems with objective functions \(c_1\) and
\(c_2\), respectively. We say that there is an \(L\)-reduction from \(A\) to \(B\) if there
exist constants \(\alpha,\beta>0\) such that:
\begin{enumerate}
    \item for each instance \(I_1\) of \(A\), one can compute in polynomial time an
    instance \(I_2\) of \(B\);
    \item if \(I_1\) is an instance of \(A\) and \(I_2\) is the corresponding instance of \(B\),
    then
    \[
    opt(I_2)\le \alpha\,opt(I_1);
    \]
    \item given any feasible solution \(s_2\) for \(I_2\), one can compute in polynomial
    time a feasible solution \(s_1\) for \(I_1\) such that
    \[
    \bigl|opt(I_1)-c_1(s_1)\bigr|
    \le
    \beta\,\bigl|opt(I_2)-c_2(s_2)\bigr|.
    \]
\end{enumerate}
\end{definition}
For maximization problems, we say that an algorithm is an \(r\)-approximation, with
\(r\ge 1\), if it always returns a feasible solution of value at least \(opt(I)/r\). We will use the following theorem, which is equivalent to Theorem~16.5 in
\cite{williamson2011design}.


\begin{theorem}\label{theo:app}
Suppose there is an \(L\)-reduction with parameters \(\alpha\) and \(\beta\) from a
maximization problem \(A\) to a maximization problem \(B\). If there is an
\(r\)-approximation algorithm for \(B\), then there is a
$
\frac{1}{1-\alpha\beta\left(1-\frac{1}{r}\right)}$-approximation algorithm for \(A\).
\end{theorem}

We begin by describing the reduction. Let $M = \{C_1, C_2, \ldots, C_m\}$ be an instance of the Max-2B3DM problem, where each $C_i$ is an ordered triple in $X \times Y \times Z$. Fix a constant $t\geq 8$. We construct a bipartite graph  $G_{M,t,3}=(V_1,V_2,E)$, as follows:
\begin{flalign*}
V_1= & X\cup Y \cup Z  \cup\{a_i \mid i \in [m] \} \\
V_2= &\bigl\{c_i \mid  i \in [m] \bigr\} \cup \bigl\{h_{i,j} \mid  i \in [m],  j \in [t-5] \bigr\} \\
E= & \Bigl\{(c_i,a_i)|\,\,i \in [m] \Bigr\} \bigcup \Bigl\{(a_i,h_{i,j}) | \,\, i \in [m],  j \in [t-5] \Bigr\} \bigcup_{i\in [m]} E_{C_i}
\end{flalign*}

where $E_{C_i}=\{(c_i, x), (c_i, y), (c_i, z)\}$ for each triple  $C_i=(x,y,z)$ in $M$. As an example see Figure  \ref{fig1}.

\begin{figure}[ht]
    \centering
\begin{tikzpicture}[
  vertex/.style={draw,circle,minimum size=6.2mm,inner sep=0pt,line width=0.9pt},
  edge/.style={line width=0.9pt}
]

\node[vertex,label=below:$x_1$] (x1) at (0.0,0.0) {};
\node[vertex,label=below:$x_2$] (x2) at (1.5,0.0) {};
\node[vertex,label=below:$y_1$] (y1) at (3.2,0.0) {};
\node[vertex,label=below:$y_2$] (y2) at (5.0,0.0) {};
\node[vertex,label=below:$z_1$] (z1) at (6.8,0.0) {};
\node[vertex,label=below:$z_2$] (z2) at (8.6,0.0) {};

\node[vertex,label=left:$c_1$] (c1) at (1.0,1.8) {};
\node[vertex,label=left:$c_2$] (c2) at (3.8,1.8) {};
\node[vertex,label=left:$c_3$] (c3) at (7.5,1.8) {};

\node[vertex,label=left:$a_1$] (a1) at (1.0,3.4) {};
\node[vertex,label=left:$a_2$] (a2) at (3.85,3.35) {};
\node[vertex,label=left:$a_3$] (a3) at (7.45,3.3) {};

\node[vertex,label=above:$h_{1,1}$] (h11) at (-0.05,4.55) {};
\node[vertex,label=above:$h_{1,2}$] (h12) at (1.05,4.55) {};
\node[vertex,label=above:$h_{1,3}$] (h13) at (2.15,4.55) {};

\node[vertex,label=above:$h_{2,1}$] (h21) at (3.25,4.48) {};
\node[vertex,label=above:$h_{2,2}$] (h22) at (4.00,4.48) {};
\node[vertex,label=above:$h_{2,3}$] (h23) at (4.95,4.48) {};

\node[vertex,label=above:$h_{3,1}$] (h31) at (6.55,4.4) {};
\node[vertex,label=above:$h_{3,2}$] (h32) at (7.55,4.4) {};
\node[vertex,label=above:$h_{3,3}$] (h33) at (8.45,4.4) {};

\draw[edge] (a1) -- (h11);
\draw[edge] (a1) -- (h12);
\draw[edge] (a1) -- (h13);
\draw[edge] (a1) -- (c1);

\draw[edge] (a2) -- (h21);
\draw[edge] (a2) -- (h22);
\draw[edge] (a2) -- (h23);
\draw[edge] (a2) -- (c2);

\draw[edge] (a3) -- (h31);
\draw[edge] (a3) -- (h32);
\draw[edge] (a3) -- (h33);
\draw[edge] (a3) -- (c3);

\draw[edge] (c1) -- (x1);
\draw[edge] (c1) -- (y1);
\draw[edge] (c1) -- (z1);

\draw[edge] (c2) -- (x2);
\draw[edge] (c2) -- (y1);
\draw[edge] (c2) -- (z1);

\draw[edge] (c3) -- (x1);
\draw[edge] (c3) -- (y2);
\draw[edge] (c3) -- (z2);
\end{tikzpicture}
\caption{The graph $G_{M,8,3}$  obtained when the instance of  Max-2B3DM  problem is  the set $M=\{(x_1, y_1, z_1), \,\,(x_2, y_1, z_1)\,\,(x_1, y_2, z_2)\}$.}
\label{fig1}
\end{figure}

\begin{claim}\label{claim}
The graph $G_{M,t,3}$ can be constructed in polynomial time, is bipartite and has maximum degree $t-4$. 
\end{claim}
\begin{proof}
It is easy to see that $G_{M,t,3}$ can be constructed in polynomial time.  Then by construction  $V_1$ and $V_2$ form a bipartition as there are no edges within the sets $V_1$ and $V_2$. Furthermore, for every $i \in [m]$, the vertex $c_i$ has degree 4, $a_i$ has degree $t-4$, and for every $j \in [t-5]$, the vertex $h_{i,j}$ has degree 1. Finally, every vertex $u \in X \cup Y \cup Z$ has degree 2, by the definition of the Max-2B3DM problem. 
\end{proof} 

We need the following lemma.
\begin{lemma}\label{lem:tre}  In the graph $G_{M,t,3}$, all $3$-clubs are of size at most $t$, with $t\geq 8$, and the only $3$-clubs of size exactly $t$ are of the form $N[c_i] \cup \{h_{i,j} \mid j \in [t-5]\}$ for $1 \leq i \leq m$.
\end{lemma}
\begin{proof}
To prove this lemma, we will show that any $3$-club in $G_{M,t,3}$ either contains exactly $t$ vertices or has at most 7 vertices (with $7 < t$). Clearly the subgraphs of $G_{M,t,3}$ induced by $N[c_i]\cup \{h_{i,j} |\,\, j \in [t-5] \}$ for $1\leq i\leq m$ are $3$-clubs of  size $t$.  To simplify the analysis, we classify the vertices in $G_{M,t,3}$ into four distinct sets: the vertices in $W = X \cup Y \cup Z$, the vertices in $C = \{c_1, \ldots, c_{m}\}$, the vertices in $A = \{a_1, \ldots, a_{m}\}$, and the vertices in $H = \{h_{i,j} \mid i \in [m],\, j \in [t-5]\}$.

Let $S$ be a  $3$-club in $G_{M,t,3}$. We will now show that if $S$ contains at least two vertices from $C$ then  $|S|\leq 7<t$. We make the following considerations regarding the composition of the vertices in $S$.
  
\begin{itemize} 
\item {\em the number of vertices in $S \cap H$ is zero.} Note that such $3$-club cannot include any vertices from $H$, as a vertex in $H$ would be at a distance of at least 4 from one of the two vertices in $C$. 

\item {\em the number of vertices in  $S \cap C$ is at most  $3$.} Suppose on the contrary there are $4$ vertices from $C$, say  $c_1,c_2,c_3$ and  $c_4$, For $S$ to be a $3$-club, there must be a vertex from $S\cap W$ adjacent to each pair of the $c_i$ ($i\in [4]$) vertices.  Let  $u$ be the vertex $S \cap W$ adjacent to  $c_1$ and  $c_2$, and  $v$ be the vertex from $S\cap W$ adjacent   $c_3$ and  $c_4$. The vertices $u$ and  $v$ are at distance at least $4$ in  $G_{M,t,3}[S]$ because $N[u] \cap N[v] = \emptyset$ as all vertices from $W$ have degree $2$.  This contradicts the assumption that $S$ is a $3$-club.

\item {\em The number of vertices in  $ S\cap W $ is at most $4$.} 
Suppose on the contrary that there are 5 such vertices, $U = \{u_1, u_2, u_3, u_4, u_5\}$. At least one vertex $u \in U$ must have degree 1 in $G_{M,t,3}[S]$. Otherwise, if all vertices in $U$ had degree 2, there would be 10 edges between these 5 vertices and the vertices in $S \cap C$. However, as shown in the previous point, there can be at most 3 vertices in $S \cap C$, and since each of these vertices has degree 3, this leads to a total of only 9 edges between $U$ and $S \cap C$. Finally notice that $u$ is at distance at least four from one of the vertices $S \cap C$  (recall that in $C$ there are at least two vertices)  contradicting the fact that $S$ is a $3$-club.
\end{itemize}
Given the above, if $S \cap A = \emptyset$ then trivially the $3$-club cannot have more than $7$ vertices. Now assume that  $S \cap A \neq \emptyset$.  Notice that it must hold $|S \cap A| = 1$, since any two  vertices in $S \cap A$  would be at distance at least $4$ from each other.  Let  $a_i \in S \cap A$, then the vertices in $W$ that $a_i$ can reach with a path of length at most three are the ones adjacent to $c_i$. Thus $|S\cap W| \leq 3$ and the total number of vertices in the $3$-club remains bounded by $1+3+3=7$. 
\end{proof}

\begin{theorem}\label{theo:APX3} 
Let $t$ be a constant with $t \geq 8$. It is NP-hard to approximate the solution of \maxdcc{t}{3}  within a factor of $\frac{95}{94}$, even for bipartite graphs with a constant maximum degree of $t-4$.
\end{theorem}
\begin{proof}
Let $G_{M,t,3}$ be the graph obtained from an instance $M$ of Max-2B3DM, for a given fixed $t$, with $t\geq 8$.  By Claim~\ref{claim} we have that  $G_{M,t,3}$ can be constructed in polynomial time, has maximum degree of $t-4$ and is bipartite. We prove now that from a matching of $k\geq 1$ triples in $M$, we can always obtain, in polynomial time, a disjoint cover with $(t,3)$-clubs that covers $t\cdot k$ vertices in  $G_{M,t,3}$. Moreover, for every disjoint cover of $G_{M,t,3}$ with  $(t,3)$-clubs that covers $k$ vertices in $G_{M,t,3}$ it is possible to obtain in polynomial time a matching of $\frac{k}{t}$  triples in $M$ (note that by Lemma \ref{lem:tre} we have that $k$ is a multiple of $t$).

Let $S=\{C_1,\ldots C_k\}$ be a subset of $k$ triples from $M$ that correspond to a matching. We define for each $i \in [k]$ the following set in $G_{M,t,3}$.
$$
S_i'=\{x,y,z,c_i,a_i, h_{i,1}\ldots h_{i,t-5}\}
$$
Notice that $S_i'$ is obviously a  $3$-club of size $t$ and, since the triples  in $S$ form a matching, it follows that $S_i' \cap S_j' =\emptyset$ for any $i\neq j$. Thus, $S_1', S_2', \ldots, S_k'$ form a disjoint cover with $(t,3)$-clubs that covers $t\cdot k$ vertices of $G_{M,t,3}$.

Let  now  $S_1, S_2, \ldots, S_r$ be a disjoint cover with $(t,3)$-clubs that covers $k$ vertices in $G_{M,t,3}$. By Lemma~\ref{lem:tre}, we know that this cover is the union of  $\frac{k}{t}$ disjoint $3$-clubs, each containing $t$ vertices of the form $N[c_i]\cup \{h_{i,j} |\,\, j \in [t-5] \}$ for some $i$, $1\leq i\leq |M|$ where $c_i$ corresponds to a triple   $(x,y,z)$ in $M$.  To obtain the  $\frac{k}{t}$ triples in $M$ corresponding to a matching, it is enough to take the triples corresponding to the  $\frac{k}{t}$ disjoint $3$-clubs from $S$.

We therefore have an $L$-reduction with $\alpha=t$ and $\beta=\frac{1}{t}$ from  Max-2B3DM to  \maxdcc{t}{3}. 
In \cite{chlebik2006}, it was proven that for Max-2B3DM, achieving an approximation factor better than $\frac{95}{94}$ is NP-hard. 
Therefore we deduce the same result for \maxdcc{t}{3}  (see \cite{williamson2011design} for more information about $L$-reductions and in particular Theorem~\ref{theo:app}. 
\end{proof}

\subsection{The \maxdcc{t}{2}  problem}\label{subsec:max_t2}
In \cite{Dondi2019} it was proven that $MAX$-$DCC(3,2)$ is APX-hard and  the complexity of $Max$-$DCC(2,2)$ was posed as an open problem. In this section we prove that $MAX$-$DCC(t,2)$ is APX-hard for any fixed $t \geq 5$, even in bipartite graphs and that $MAX$-$DCC(2,2)$ can be solved in polynomial time.

In order to prove that $DCC(t,2)$ problem with $t\geq 5$  is APX-hard, we present an $L$-reduction from Max- 2B3DM. We begin by describing the reduction. Let $M=\{C_1,C_2\,\ldots C_m\}$ be an instance of Max- 2B3DM where each $C_i$ is an ordered triple in $X\times Y\times Z$.  

Fix a constant $t$ with $t\geq 5$. We construct a bipartite graph  $G_{M,t,2}=(V_1,V_2,E)$,  as follows:

\begin{flalign*}
V_1= & X\cup Y \cup Z  \cup \{h_{i,j}| 1\in [m],\,\,j \in  [t-4]  \}\\
V_2= &\{c_i| 1\leq i\leq |M|\} \\
E= & \bigl\{(c_i,h_{i,j})|\,\,i \in [m]\,\,j \in [t-4] \bigr\} \bigcup_{i\in [m]} E_{C_i}
\end{flalign*}
where $E_{C_i}=\bigl\{(c_i,\,x), \, (c_i,\,y),\, (c_i,\,z)\bigr\}$ for each triple  $C_i=(x,y,z)$ in $M$.

As an example see Figure  \ref{fig2}.

\begin{figure}[t]
    \centering
    \begin{tikzpicture}[
      vertex/.style={draw,circle,minimum size=6.2mm,inner sep=0pt,line width=0.9pt},
      edge/.style={line width=0.9pt}
    ]

    \node[vertex,label=below:$x_1$] (x1) at (0.0,0.0) {};
    \node[vertex,label=below:$x_2$] (x2) at (1.6,0.0) {};
    \node[vertex,label=below:$y_1$] (y1) at (3.2,0.0) {};
    \node[vertex,label=below:$y_2$] (y2) at (4.8,0.0) {};
    \node[vertex,label=below:$z_1$] (z1) at (6.4,0.0) {};
    \node[vertex,label=below:$z_2$] (z2) at (8.0,0.0) {};

    \node[vertex,label=left:$c_1$] (c1) at (0.8,1.6) {};
    \node[vertex,label=above left:$c_2$] (c2) at (3.9,1.6) {};
    \node[vertex,label=above left:$c_3$] (c3) at (7.0,1.6) {};

    \node[vertex,label=above:$h_{1,1}$] (h11) at (0.8,3.0) {};
    \node[vertex,label=above:$h_{2,1}$] (h21) at (3.95,2.95) {};
    \node[vertex,label=above:$h_{3,1}$] (h31) at (6.95,2.9) {};

    \draw[edge] (h11) -- (c1);
    \draw[edge] (h21) -- (c2);
    \draw[edge] (h31) -- (c3);

    \draw[edge] (c1) -- (x1);
    \draw[edge] (c1) -- (y1);
    \draw[edge] (c1) -- (z1);

    \draw[edge] (c2) -- (x2);
    \draw[edge] (c2) -- (y1);
    \draw[edge] (c2) -- (z1);

    \draw[edge] (c3) -- (x1);
    \draw[edge] (c3) -- (y2);
    \draw[edge] (c3) -- (z2);

    \end{tikzpicture}
\caption{The graph $G_{M,5,2}$  obtained when the instance of  Max-2B3DM  problem is  the set $M=\{(x_1, y_1, z_1), \,\,(x_2, y_1, z_1)\,\,(x_1, y_2, z_2)\}$.}
\label{fig2}
\end{figure}

\begin{claim}\label{claim2}
The graph $G_{M,t,2}$ can be constructed in polynomial time, is bipartite and has maximum degree $t-1$. 
\end{claim}
\begin{proof}
It is easy to see that $G_{M,t,2}$ can be constructed in polynomial time.  Then by construction  $V_1$ and $V_2$ form a bipartition as there are no edges within the sets $V_1$ and $V_2$. Furthermore, for every $i \in [m]$, the vertex $c_i$ has degree $t-1$ and for every $j \in [t-4]$, the vertex $h_{i,j}$ has degree 1. Finally, every vertex $u \in X \cup Y \cup Z$ has degree 2, by the definition of the Max-2B3DM problem. 
\end{proof}

\begin{lemma}\label{due} In the graph $G_{M,t,2}$, all the $2$-clubs are of size at most $t$, with $t\geq 5$ and the only $2$-clubs of size exactly $t$ are of the form  $N[c_i]$ for $1\leq i\leq |M|$.
\end{lemma}
\begin{proof}
To prove the claim, we will show that any $2$-club in $G_{M,t,2}$ either contains exactly $t$ vertices or has at most $4$ vertices (with $4<t$). 
Clearly the subgraphs   of $G_{M,t,2}$ induced by  $N[c_i]$ for $1\leq i\leq m$ are $2$-clubs of  size $t$. \\
To simplify the analysis, we classify the vertices in $G_{M,t,2}$ into three distinct sets: the vertices in $W= X\cup Y\cup Z$, the vertices in $C=\{c_1\ldots c_{m}\}$ and the vertices in 
 $H= \{h_{i,j}|  i\leq [m],\,\, j\in [t-4] \}$.\\
Let $S$ be a $2$-club  in $G_{M,t,2}$. We will now show that  if $S$ contains at least two vertices from $C$ then $|S|\leq 4<t$. \\
Note that  the number of vertices in $S\cap H$ is zero as a vertex in $H$ of $G_{M,t,2}[S]$ would be at distance of at least $3$ from one of the two vertices in $C$.  Furthermore, in $G_{M,t,2}[S]$ each vertex in $S\cap C$   must be adjacent to each vertex in $S\cup W$.  Consequently,   $S$ forms a complete bipartite graph having $S\cap C$ and $S\cap W$ as independent sets. The  degree of the vertices in $S\cap W$ is at most two thus $S$  can have at most   $4$ vertices (two in $S\cap C$ and two in $S\cap W$). 
\end{proof}

\begin{theorem}
\label{theo:APX2}
Let   $t$ a constant with , $t\geq 5$. It is  NP-hard to approximate the solution of \maxdcc{t}{2}  to within $\frac{95}{94}$ even for bipartite graphs of degree at most $t-1$.  
\end{theorem}
\begin{proof}
Let $G_{M,t,2}$ be the  graph obtained from an instance $M$ of Max-2B3DM, for a given fixed $t$, with $t\geq 5$. By Claim~\ref{claim2} we have that $G_{M,t,2}$ can be constructed in polynomial time, has a maximum degree of $t-1$ and is bipartite.

\noindent We prove now  that from a matching of $k\geq 1$ triples in $M$, we can always obtain in polynomial time, a disjoint cover with $(t,2)$-clubs that cover $t\cdot k$ vertices in $G_{M,t,2}$. Moreover, for every disjoint cover of $G_{M,t,2}$ it is possible to obtain in polynomial time a matching of $\frac{k}{t}$ triples in $M$ (note that by Lemma \ref{due} we have that $k$ is a multiple of $t$).\\

Let $S=\{C_1,\ldots C_k\}$ be a subset of $k$ triples from  $M$ that correspond to a matching. We define for each $i\in [k]$ the following set in $G_{M,t,2}$
$$S'_i=\{x,y,z,c_i,h_{i,1}\ldots h_{i,t-4}\}$$
Note that $S'_i$ is obviously a $2$-club of size $t$ and, since the triples in $S$ form a matching, it follows that $S'_i\cap S'_j =\emptyset$ for any $i\neq j$. Thus $S'_1,S'_2,\ldots S'_k$ form a disjoint cover with $(t,2)$-clubs that cover $t\cdot k$ vertices of $G_{M,t,2}$.

Let  now  $S_1,S_2,\ldots S_r$ be a  disjoint cover with $(t,2)$-clubs that cover $k$ vertices in $G_{M,t,2}$. By Lemma \ref{due}, we know that this cover  is the union of   $\frac{k}{t}$ disjoint $2$-clubs, each  containing $t$ vertices  of the form   $N[c_i]$. for some $i, 1\leq i\leq |M|$.  where $c_i$ corresponds to a triple $(x,y,z)$ in $M$. To obtain the $\frac{k}{i}$ triples in $M$ corresponding to a matching, is enough to take the triples corresponding to the $\frac{k}{i}$ To obtain the    $\frac{k}{t}$ triples it is enough to take the triples corresponding to the  $\frac{k}{t}$ disjoint $2$-clubs from $S$.

We therefore have an $L$-reduction with $\alpha=t$ and $\beta=\frac{1}{t}$ from  Max-2B3DM to  \maxdcc{t}{2}. 
In \cite{chlebik2006}, it was proven that for Max-2B3DM, it is NP-hard to achieve an approximation factor of $\frac{95}{94}$. Thus, from Theorem~\ref{theo:app}, we deduce the same result for \maxdcc{t}{2}.
\end{proof}

In \cite{Dondi2019}, it is shown that for any  $s\geq 3$ \maxdcc{2}{s}  is solvable in linear time. The authors leave the problem of determining the complexity of \maxdcc{2}{s}  as an open problem. We prove the following ideas similar to the one in \cite{Dondi2019}.\\

\begin{theorem} \maxdcc{2}{2}   can be solved in linear time. 
\end{theorem}
\begin{proof}
Let $G$ be the input graph for \maxdcc{2}{2} . We will prove the claim by describing a linear algorithm that produces a cover of disjoint $(2,2)$-clubs that covers the maximum number of vertices of $G$. Consider the graph $G'$, obtained by  removing the isolated vertices from $G$. Note that isolated vertices cannot be covered by a $(2,2)$-club, so they can be disregarded.
We can assume that  $G'$ is connected as otherwise, we apply the following procedure  to each connected component. 

From $G'$, construct a rooted spanning tree $T_{G'}$.  This can be done in linear time see \cite{cormen01introduction}.  Notice that $T_{G'}$ has a height of at least one as $G'$ has no isolated vertices. Thus, let $x$ be a vertex of $T_{G'}$ at maximum level and let $y$ be the parent of $x$. Consider the subtree $T_y$ rooted at $y$. Clearly, $T_y$ contains at least $2$ vertices. Moreover, as $x$ has maximum level then $T_y$ is of height 1 and all vertices in this subtree are at a distance at most $2$ from each other. Hence, the set $S_1$ of vertices in $T_y$ form a  $(2,2)$-club. We add $S_1$ to the solution and remove $T_y$ from $T_{G'}$. The remaining tree is still connected, so we can iterate the process until either we reach an empty tree or we reach a tree consisting of a single vertex. In the latter case, let $s$ be this vertex and let $S_i$ be the $(2,2)$-club added to the solution by the last iteration. Notice that $s$ is at distance 1 from the root of this subtree and at distance 2 from any vertex of $S_i$. Therefore, we can add $s$  to $S_i$ and still have a $(2,2)$-club. In the end  the solution obtained is a cover of \emph{all} the vertices of $G'$ with disjoint $(2,2)$-clubs and thus is a  solution for \maxdcc{2}{s} on $G$.
 \end{proof}

\section{Open problems}\label{sec:conslusions}
Several problems remain open in the context of bipartite graphs. Specifically, the cases of \pc{t}{4} and \pc{t}{6} for any $t\geq 2$ remain unsolved.

Regarding the second problem, for $s=3$, the complexity of the M\maxdcc{t}{3}  problem is still open for $3 \leq t \leq 7$. Similarly, for $s=2$, the case of \maxdcc{4}{2}  remains unsolved.

\section*{Acknowledgements}
The authors thank the reviewers for carefully reading of the paper and in particular Reviewer 2 for identifying a flaw in the previous proof of the complexity of \pc{t}{s}, which gave us the opportunity to reconsider the argument and strengthen the result.

This research was supported in part by MUR PRIN Project EXPAND, grant number 2022TS4Y3N.

\bibliographystyle{plain}
\bibliography{myreferences}

\end{document}